\documentclass[a4paper,UKenglish,cleveref,autoref,thm-restate]{lipics-v2021}
\usepackage{blkarray}
\usepackage{amsmath}
\usepackage{xfrac} 
\usepackage[noend]{algpseudocode}
\usepackage{algorithm}
\usepackage{complexity}
\newcommand{\mcPo}{\mcP^{(1)}}
\newcommand{\mcPt}{\mcP^{(2)}}
\newcommand{\mcPi}{\mcP^{(i)}}
\newcommand{\mcQo}{\mcQ^{(1)}}

\newcommand{\mcQi}{\mcQ^{(i)}}

	\newcommand{\bbF}{\mathbb{F}}

	\newcommand{\bbN}{\mathbb{N}}

	\newcommand{\bbZ}{\mathbb{Z}}

\newcommand{\mcA}{\mathcal{A}}	\newcommand{\mcB}{\mathcal{B}}
\newcommand{\mcC}{\mathcal{C}}	
	\newcommand{\mcF}{\mathcal{F}}
\newcommand{\mcG}{\mathcal{G}}	\newcommand{\mcH}{\mathcal{H}}

\newcommand{\mcM}{\mathcal{M}}	\newcommand{\mcN}{\mathcal{N}}
\newcommand{\mcO}{\mathcal{O}}	\newcommand{\mcP}{\mathcal{C}}
\newcommand{\mcQ}{\mathcal{C}'}

	\newcommand{\bmF}{\boldsymbol{F}}

	\newcommand{\bmL}{\boldsymbol{L}}
\newcommand{\bmM}{\boldsymbol{M}}

	\newcommand{\bmp}{\boldsymbol{p}}

	\newcommand{\bmv}{\boldsymbol{v}}

   \newcommand{\sF}{{\mathscr{F}}}

	\newcommand{\mfR}{\mathfrak{R}}

\newcommand{\wt}{\widetilde}

\newcommand{\vA}{\vec{A}}

\newcommand{\vI}{\vec{I}}	\newcommand{\vJ}{\vec{J}}

\newcommand{\abs}[1]{\left|#1\right|}                  

\renewcommand{\A}{\mcA^{(1)}}
\newcommand{\B}{\mcA^{(2)}}
\newcommand{\N}{\mcN^{(1)}}
\renewcommand{\M}{\mcN^{(2)}}
\newcommand{\aA}{A^{(1)}}
\newcommand{\bB}{A^{(2)}}
\newcommand{\sa}{a^{(1)}}
\newcommand{\sab}{a^{(2)}}
\newcommand{\n}{n^{(1)}}
\newcommand{\m}{n^{(2)}}

\renewcommand{\epsilon}{\varepsilon}

\renewcommand{\epsilon}{\varepsilon}
\newcommand{\ignore}[1]{}

\newtheorem{property}[theorem]{Property}

\title{A parallel algorithm for the odd two-face shortest $k$-disjoint path problem} 

\author{Srijan Chakraborty}{Chennai Mathematical Institute}{srijanc@cmi.ac.in}{}{}

\author{Samir Datta}{Chennai Mathematical Institute}{sdatta@cmi.ac.in}{}{}

\authorrunning{S. Chakraborty and S. Datta} 

\Copyright{Jane Open Access and Joan R. Public} 



\category{} 

\relatedversion{} 

\acknowledgements{I want to thank \dots}

\nolinenumbers 

\EventEditors{John Q. Open and Joan R. Access}
\EventNoEds{2}
\EventLongTitle{42nd Conference on Very Important Topics (CVIT 2016)}
\EventShortTitle{CVIT 2016}
\EventAcronym{CVIT}
\EventYear{2016}
\EventDate{December 24--27, 2016}
\EventLocation{Little Whinging, United Kingdom}
\EventLogo{}
\SeriesVolume{42}
\ArticleNo{23}

\begin{document}

\maketitle

\begin{abstract}
The shortest Disjoint Path problem ($\SDPP$) requires us to find pairwise
vertex disjoint paths between $k$ designated pairs of terminal vertices such that the sum of the path lengths is minimum. 
The focus here is on $\SDPP$ restricted to planar graphs where all terminals 
are arbitrarily partitioned over two distinct faces with the 
additional restriction that each face is required to contain an odd number 
of terminals. We call this problem the Odd two-face planar $\SDPP$.
It is shown that this problem is solvable in randomized 
polynomial time and even in $\RNC$.
This is the first parallel (or even polynomial time) solution for
the problem.

Our algorithm combines ideas from the randomized solution for 
$2$-$\SDPP$ by Bj\"orklund and 
Huslfeldt with its parallelization by Datta and Jaiswal along with 
the deterministic algorithm for One-face planar $\SDPP$ by Datta, Iyer, Kulkarni
and Mukherjee. 

The proof uses a combination of two involutions
to reduce a system of linear equations modulo a power of $2$ to a system of 
triangular form that is, therefore, invertible. This, in turn,
 is proved by showing that the
matrix of the equations, can be interpreted as (the adjacency matrix of)
 a directed acyclic graph (DAG).
While our algorithm is primarily algebraic the proof remains combinatorial. 

We also give a parallel algorithm for the $(A+B)$-$\SDPP$ introduced by Hirai 
and Namba.

\end{abstract}

\section{Introduction}
\label{sec:intro}
Let $G$ be a graph and $\vec{S} = (s_1,\ldots,s_k), \vec{T} = (t_1,\ldots,t_k)$
be two sequences of vertices, all distinct, the Disjoint Path Problem 
($\DPP$) is to determine if there are pairwise vertex disjoint paths\footnote{For us $\DPP$ is
synonymous with \emph{vertex} disjoint paths and we will not have occasion to delve 
into the extant, large literature on \emph{edge} disjoint paths.
Also notice that $\DPP$ is more specific than finding
disjoint paths between two \emph{sets} of vertices $S,T$. Our problem is distinct from disjoint shortest paths which has a polynomial time algorithm \cite{Lochet}.}
between the terminals $s_i,t_i$ for each $i \in \{1,\ldots,k\}$. 
\subparagraph*{Motivation}
$\DPP$ and its variants have played an important role in theoretical computer
science with applications ranging from network routing \cite{ORS,SM}, VLSI design
\cite{AKW} to the celebrated Graph Minor project of Robertson and Seymour
\cite{RSXIII}. It was also one of the problems that figured in Karp's classic
paper \cite{Karp} on $\NP$-completeness. 
The planar restrictions of
$\DPP$ remain of current interest to the parameterized community \cite{LMPSZ,WZ,COO}.
$\DPP$ can also be cast as an optimization problem and approximation algorithms
for it form an active area of research \cite{CKN,CKN2}. 
\subparagraph*{Related Work}
$\DPP$ is 
$\NP$-complete for both undirected and directed graphs \cite{Karp,HP}. Somewhat more
restrictive is the $k$-$\DPP$ where the terminal sequences are of length 
precisely $k$.
The directed problem remains $\NP$-complete even for $k=2$, \cite{FHW}. For
undirected graphs, 
Robertson and Seymour \cite{RSXIII} proved that the problem is in cubic time with 
a galactic dependence on $k$. More recently, Kawarabayashi et al \cite{KKR}
have reduced the time to quadratic from cubic. 

Meanwhile, for certain subclasses of directed graphs viz.
planar directed graphs, polynomial time algorithms (with exponents depending 
on $k$) are known \cite{Schrijver} for $k$-$\DPP$. Based on these
parameterized algorithms have been developed for various variants of planar
problems \cite{Pilipczuk,LMPSZ}.

We will be concerned with the problem of determining if there is a $k$-$\DPP$
solution with sum of weights at most a given value. In other words we will be
interested in solving the shortest $k$-$\DPP$ or $\SDPP$.
 For $k > 2$ it is not known if
the problem tractable in (randomized) polynomial time or not. For $k = 2$,
Bj\"orklund and Husfeldt \cite{BH}  give a randomized polynomial time algorithm
for the problem based on computing a univariate polynomial permanent modulo
$4$. Notice that permanent modulo $2$ is identical to determinant modulo $2$
since the only difference between determinant and permanent is the sign and 
${{-1 \equiv 1}}\bmod{2}$. This can be lifted to yield the value of
the permanent of a univariate polynomial matrix modulo $4$ in polynomial time.
Using the isolation lemma \cite{MVV} it was shown in \cite{BH} that 
finding the least degree monomial in the linear combination of three 
permanents modulo $4$ yields a method to infer the weight of $2$-$\SDPP$. 


The above technique was based on a thread originating from Valiant \cite{Valiant} who showed that
the permanent of an integer matrix modulo $4$ can be computed in polytime 
which was refined to $\ParityL \subseteq \NC^2$ in \cite{BKR}. This technique
was extended to univariate polynomials in \cite{DJ} to yield the first 
(randomized) parallel algorithm for univariate permanent modulo $2^t$ for a 
constant $t$. Combined with \cite{BH} this yielded the first (randomized)
parallel algorithm for $2$-$\SDPP$.

In a different thread, \cite{DIKM} showed that $k$-$\SDPP$ in
planar graphs where all terminals lie on a single face can be reduced to
the sum of $4^k$ univariate determinants. The basic idea was to add ``demand''
edges inside the one face containing terminals and compute a generating function for cycle covers where each cycle cover is in correspondence with some $k$ disjoint paths (though not always in the desired pairing). This leads to $4^k$ linear equations with the right hand 
side being the sum of cycle covers compatible with the applied demand edges 
which can therefore can be evaluated as a determinant. The left hand side
consisted of the sum of various types of $k$ disjoint paths. It was noticed
in \cite{DIKM} that the left hand side gave rise to a square matrix that is
upper triangular and so can be inverted to supply the generating function 
of the $k$ disjoint paths of various types.

This was generalised to one-face $k$-$\SDPP$ with one extra arbitrarily placed 
terminal, by  Kobayashi and Terao \cite{KT}. 
They generalised the technique from Hirai and Namba \cite{HiraiNamba} who 
gave an alternate proof of Bj\"orklund and Husfeldt \cite{BH} based on computing Hafnian modulo $4$.

The algorithms from \cite{BH,DIKM} are algebraic while there has been concurrent
combinatorial work on similar shortest path problems. Notably, 
\cite{Borradaile} showed
that the ``sequential'' version of the one-face problem is solvable in 
polytime using a flow based algorithm. Similarly when the sources are on one 
face and sinks are on the other, the shortest disjoint path problem can be
solved in polynomial time using flows \cite{deVerdierSchrijver}.
Similarly when the $k$ pairs are distributed arbitrarily across two distinct
faces but $k \leq 2$ then the problem is solved in \cite{KobayashiSommer}.

\subparagraph*{Results}
We focus on the problem of finding shortest $k$-disjoint paths between $2k$ terminals 
distributed arbitrarily amongst two faces under the proviso
 that each face contains an odd 
number of terminals. We call the problem ``Odd two face $k$-$\SDPP$'' and show
the following:
\begin{theorem}\label{theorem:1}
	The decision as well as the search version of the Odd Two-Face $k$-$\SDPP$ is in randomized $\NC^2$ for constant $k$.
\end{theorem}
Our result subsumes and  improves all known results on decision and search
 of one-face and two-face\footnote{If all the paths are across the two faces -- the case
 	studied in \cite{deVerdierSchrijver} and even in number then apparently our result doesn't 
 	subsume it. However, see the footnote to \cref{algo:1}.}  planar cases \cite{KobayashiSommer,deVerdierSchrijver,DIKM,KT} and parallelises
all except \cite{DIKM}
which is already parallel.

We also solve a problem $(A+B,q)$-$\SDPP$ that generalises and parallelises the $(A+B)$-$\SDPP$ from \cite{HiraiNamba}. Let $A,B$ be a set of terminals, with $|A|=k_1,|B|=k_2$ and $k_1+k_2=k$. The goal is to find the shortest disjoint paths where exactly $q$ paths have one endpoint in $A$ and one in $B$, and the rest of the paths are among $A$ terminals or $B$ terminals.
\begin{theorem}\label{theorem:2}
	An $(A+B,q)$ $\SDPP$ can be found in randomized $\NC^2$ when $|A|+|B|=k$ is a constant. In particular, when $q=0$ shortest $A+B$ path is in randomized $\NC^2$.
\end{theorem}
\subparagraph*{Proof Idea}
In \cite{DIKM}, the One-face $k$-disjoint problem was solved by writing 
a number of linear equations. Each equation corresponds to a configuration
$\mcP$,
that is, a partition of the terminals in $k$ sources and $k$-sinks. They also
need the notion of permutation of sinks $\pi$. Thus an equation coresponds
to the fact that the $k$-disjoint paths $\mcP$ can be partitioned into ones 
where the $i$-th source is mapped to the $\pi(i)$-th sink for various $\pi$'s.
 It is crucial in \cite{DIKM} that all the $k$-disjoint paths occurring in a 
configuration $\mcP$ do so with the same sign. Thus to solve the problem 
it suffices to prove that the system of equations is a square system and
further it is invertible. To prove invertibility it suffices to prove that
it is an upper triangular system. Notice that solving the system allows them
to find the number of $k$-disjoint paths for every $\pi$.

We adapt the above approach to the two face case but need to make significant changes.
Firstly, we are unable to count the number of $k$-disjoint paths corresponding 
to a configuration via a determinant because different paths have different 
signs. Thus we are forced to do the counting via permanents and side-step the
hardness thereof by working modulo a power of $2$ -- notice that from 
\cite{BH,DJ} we know that permanents of univariate polynomials modulo powers of
$2$ are tractable. We extend this result from univariate to bivariate permanents \cref{lemma:dj}. Again we prove that the concerned system of equations forms a
square system \cref{lem:c1}.

However we get a non triangular system for even small cases
(for instance, see \cref{ex:one}). Thus to show invertibility we left multiply
the system $\bmM \bmv = \bmp$ by a carefully  crafted matrix $\bmL$,
 to yield a square matrix $\bmF$. 
Ideally, we would like to ensure that in the matrix $\bmF$ the 
`undesirable' entries of $\bmM$ that is the ones that prevent $\bmM$ from being 
triangular are expunged so that $\bmF$ is triangular. 
The undesirable entries are  further classified into `bad' and `face-equivalent'
ones.  We show that in $\bmF$ the bad entries cancel out in pairs as do the
face-equivalent ones via two involutions both packed in the matrix $\bmL$
(\cref{sec:cancel}).
However the matrix $\bmF$ we obtain is not apparently triangular in the given 
ordering of its rows and columns.

To show that $\bmF$ can be made triangular by permuting rows and columns requires
another combinatorial proof (\cref{sec:inv}). 
Our algorithm and its motivation are algebraic but the proof is combinatorial.

For the proof of \cref{theorem:2} we again invert a square system of linear permanental equations.

\subparagraph*{Organization}
We start with a few preliminaries in \cref{sec:prelims} to recall the
part of previous work relevant for our paper. Next we set up notation
in \cref{sec:3} including a description of our main algorithm \cref{subsec:algo}.
In \cref{sec:proof} we present the proof of correctness of the algorithm.
In \cref{sec:abpaths} we solve the $(A+B,q)$-$\SDPP$.
Finally we conclude with \cref{sec:concl}.

\section{Preliminaries}
\label{sec:prelims}
We will have occasion to  use parallel complexity classes like $\NC, \ParityL,
\Log$ -- we refer the reader to standard texts like \cite{AroraB,Vollmer} for
those. 
For linear algebraic notions like inverse, adjoint and determinant refer to
any standard linear algebra text like \cite{HoffmanK}. Graph theoretic terms like planarity and planar dual
can be looked up in \cite{Diestel}. We need only the definition of the 
permanent of a square matrix -- which is identical to that of the determinant
except the sign of a permutation is omitted. Also see \cite{DJ} for a formal
definition of permanent and cycle cover etc. in a form that we need.

For a polynomial $p\in\bbZ[x,y],$ let $[y^d]_ap$ be the coefficient of $y^d$ in $p\in{\bbZ[x,y]}/{(y^a-1)}$.
\begin{proposition}\label{prop:extract}
	For a fixed $c\ge 1$, $[y^d]_a\perm(A)\bmod 2^c$ can be computed in $\ParityL\in\NC^2$, where $A$ is a bivariate polynomial matrix.
\end{proposition}
The proof is immediate by summing up the coefficients of exponents congruent to $d\bmod a$ using the following Lemma.
\begin{lemma}\label{lemma:dj}
	Let $c\ge1$ be fixed and $A$ be a $n\times n$ matrix of bivariate\footnote{This generalises to constantly many variables.} integer polynomials of degree $poly(n)$. We can compute the coefficients of $\perm(A) \bmod 2^c $ in $ \ParityL\subseteq \NC^2$.
\end{lemma}
\begin{proof}
	This follows from \cref{lemma:extract} since we can compute the permanent of integer matrices in $\ParityL$ using \cite{BKR}. (This can also be done using Kronecker substitution and \cite{DJ}.)
\end{proof}
\begin{lemma}[Robertson-Seymour\cite{RSvi} Rephrased]\label{lemma:rs}
	We represent the surface on which two circles $C_1,C_2$ are drawn by $\sigma=\{(r,\theta): 1\le r\le 2,0\le \theta \le 2\pi\}$. Let $f:[0,1]\rightarrow \sigma$ be a continuous path on $\sigma$ with $f(0)\in C_1,f(1)\in C_2$. Then it has finite winding number $\wnd(f)$ defined intuitively as the number of times $f$ winds around in the clockwise sense. Let $L$ be a set of disjoint paths across $C_1$ and $C_2$ drawn on $\sigma$. Then clearly, for all $p\in L$, $\wnd(p)$ is a constant and we denote this by the common value $\wnd(L)$.
\end{lemma}
\begin{lemma}[Isolation Lemma \cite{MVV}]\label{lemma:iso}
	Given a non-empty $\sF\subseteq2^{[m]}$, if one assigns for each $i\in[m]$, $w_i\in[4m]$ uniformly at random then with probability at least $3/4$, the minimum weight subset in $\sF$ is unique; where the weight of a subset $S$ is $\sum\limits_{i\in S}w_i$.
\end{lemma}
Here we describe a preprocessing step using which we can connect cycle covers with the Shortest $k$-$\DPP$.
\subparagraph*{Preprocessing:} Let $G$ be an undirected graph with $2k$ terminal vertices $s_1,\ldots,s_k$ and $t_1,\ldots t_k$, which are the sources and sinks respectively. We do the following operations on $G$ to get the new digraph graph $H$:
(1) Make all edges incident to $s_i$'s outgoing, and all edges incident to $t_i$'s incoming. 
(2) If the weight of an edge $e$ was $w_e$ in $G$, the weight of $e$ in $H$ is $x^{w_e}$. Here $x$ is an indeterminate.
(3) Add directed edges of unit weights $=x^0$ from $t_i$ to $s_i$ for all $i$ ; we call these \textit{demand edges}.
(4) Add self-loops of unit length $=x^0$ on all the non-terminal vertices.
We call the adjacency matrix of $H$ to be $A_H$. We know there is a bijection between monomials in the Permanent of $A_H$, and cycle covers in $H$. Also, note that a cycle cover in $H$ contains cycles of three types:
(1) a cycle consisting alternately of paths between some two terminals, and demand edges (2) a non-trivial cycle not containing any terminal vertex (3) a non-terminal self-loop.

Thus, every cycle cover of $H$ contains a set of $k-$disjoint paths between the terminals (not necessarily in the desired order). Further, any set of $k-$disjoint paths between the terminals can be extended to cycle covers, using the demand edges and self-loops on the unused vertices. Also, note that weights of $H$ are multiplicative, and the exponent is the sum of the corresponding weights of $G$. Thus we have the following Lemma.
\begin{lemma}
The non-zero monomials in $\perm(A_H)$ are in bijection with the cycle covers in $H$, and the cycle covers in $H$ are extensions of $k-$disjoint paths in $G$. Moreover, the bijection preserves the degree of the monomial as the length of the cycle cover it is mapped to.
\end{lemma}

    Define a \emph{path configuration} as a matching among the $2k$ terminals. Let $\mcP$ be a path configuration. Call $P$ to be an instance of $\mcP$ or equivalently $P\in\mcP$, if the $k$ edges in $\mcP$ correspond to $k$ disjoint paths among the terminals in $P$. Denote $w(P)$ as the weight of $P$, and for a cycle cover $C$ in $H\setminus P$, let $w(C)$ be the weight of the cycle cover. Let $\mcC$ be the set of cycle covers in $H\setminus P$. Then, define $h_P(x)=\sum\limits_{C\in\mcC}x^{w(P)+w(C)}$. Finally, define $h_\mcP(x)=\sum\limits_{P\in\mcP}h_P(x)$
Let $S\sqcup T,|S|=|T|$ be a partition of the set of terminals in $G$ set to sources and sinks, and $\mcM$ be the set of all \emph{Path Configurations} where every matched edge has one terminal in $S$, and the other in $T$. Then, note that
\begin{proposition}\label{fact:mat}
	 $\perm(A_H)=\sum\limits_{\mcP\in\mcM}h_{\mcP}(x)$.
	 \end{proposition}
	  Ultimately, we want $S=(s_1,\ldots,s_k),T=(t_1,\ldots,t_k)$ and the polynomial $h_{\mcP}(x)$ where $\mcP=\{(s_1,t_1),\ldots,(s_k,t_k)\}$. To achieve this we shall vary $S,T$ in our intermediate computations.

\section{Disjoint Paths on Two Faces: The Odd Case}
\label{sec:3}
In this section, we solve a restricted version of the shortest $k$-$\DPP$ on planar graphs such that all the terminals lie on two faces.
\subsection{Notation and Modification:}\label{sec:3.1}{
	 We refer to $f_1$ as the outer face and $f_2$ as the inner face. Let the set of terminals on $f_i$ be $K_i$, and $\abs {K_i}=k_i$. We label the terminals as $K_1=\{z_1^{(1)},\ldots,z_{k_1}^{(1)}\},K_2=\{z_1^{(2)},\ldots,z_{k_2}^{(2)}\}$ in the clockwise order, and use the notation $z_j^{(i)}+1=z_{(j\bmod {k_i})+1}^{(i)}$ for $i=1,2$. We further assume that each face contains an odd number of terminals.
	 
	 Terminals can be classified as sinks and sources. The number of sinks and sources on $f_1$ is $\ell_1$ and $k_1-\ell_1$ respectively. The number of sinks and sources on $f_2$ are $k_2-\ell_2$ and $\ell_2$ respectively. Set $q=k-\ell_1-\ell_2=k_1-2\ell_1=k_2-2\ell_2$, we get this equality by equating $\#$sinks$=\#$sources. Observe that the paths with both endpoints on $f_1$ must have one endpoint in a sink. Hence there are at most $\ell_1$ paths with both endpoints on $f_1$. Therefore \begin{proposition}\label{prop:q}
	 	When there are exactly $\ell_1$ sinks and $\ell_2$ sources on $f_1$ and $f_2$ respectively,
	 	the number of paths that go across the two faces is at least $q$.
	 	\end{proposition}
	 
	  For terminals $t_1\ne t_2$ on the same face ($f_1$ or $f_2$), we say $t_1\prec t_2$ when there are even number of terminals between $t_1$ and $t_2$ in the clockwise sense. Observe that it is a well-defined order since the number of terminals on either face is odd\footnote{This is the first place we assume both faces have an odd number of terminals}. We emphasise the fact that $\prec$ is not transitive.
	  
	   For a \textit{path configuration} $\mcP$, define $\mcPi$ to be the set of paths (or matchings), with both endpoints on $f_i$. For a path $p\in\mcPi$ for $i=1,2$, define $\alpha_p\prec \beta_p$ as the two endpoints of $p$. For two terminals $t_1\prec t_2$ on the same face, define the \textit{even interval} of $t_1,t_2$ by the set of terminals $\{t_1,t_1+1,\ldots,t_2\}$ in the clockwise direction.
	Call a terminal \textit{free} w.r.t a path configuration $\mcP$, if it is matched to some terminal on the other face. Define the \emph{alpha sequence} of $\mcP$ on either face for a configuration $\mcP$ as $\alpha^{(1)}(\mcP)=\{\alpha_p|p\in\mcPo\}$, and $\alpha^{(2)}(\mcP)=\{\alpha_p|p\in\mcPt\}$. Observe that there cannot be any \emph{free} vertices in the even interval of $\alpha_p,\beta_p$ for any $p\in\mcPo$ or $\mcPt$, and the paths on $\mcPi$ cannot interlace\footnote{paths $p$ and $q$ interlace if $\alpha_p,\alpha_q,\beta_p,\beta_q$ are in the clockwise order.} for $i=1,2$. Hence,  $\alpha^{(1)}(\mcP),\alpha^{(2)}(\mcP)$ uniquely determine $\beta_p$ for all the paths $p$ in $\mcPi$ for $i=1,2$. Thus $\alpha^{(1)}(\mcP),\alpha^{(2)}(\mcP)$ uniquely determine $\mcPo,\mcPt$.
\begin{figure}[hbt!]
	
	\begin{minipage}[c]{\linewidth}
		\centering
		\includegraphics[width=3.5cm]{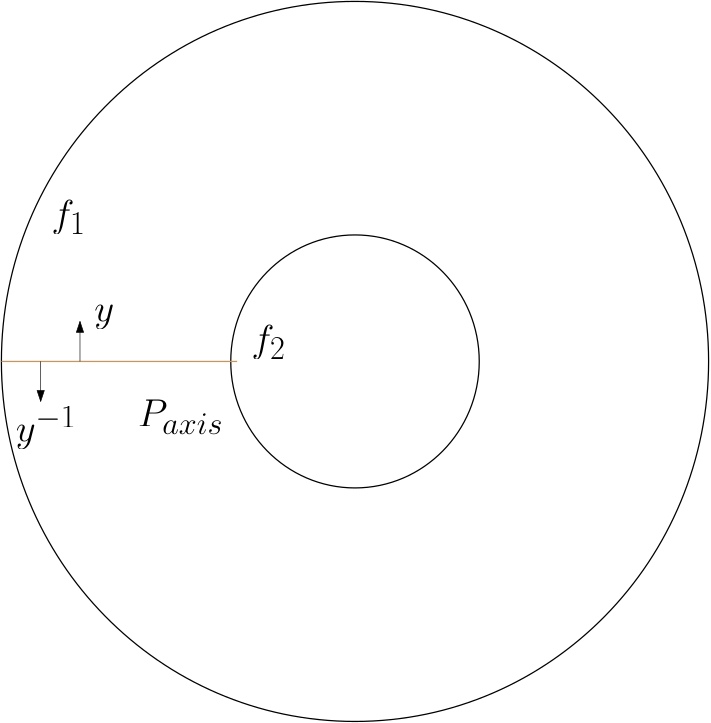}
		\caption{Clockwise edges get multiplied by $y$, and anti-clockwise edges by $y^{-1}$.}
		\label{fig:1}
	\end{minipage}
\end{figure}

Consider the bi-directed dual $G^{\ast}$, where $f_1,f_2$ are vertices. We connect $f_1$ and $f_2$ by a path  $P_{axis}$ in $G^{\ast}$. We consider the corresponding primal arcs of $P_{axis}$ which are directed in the clockwise order and multiply their weights by $y$. Similarly, multiply the primal arcs of $P_{axis}$ directed in the counter-clockwise order by $y^{-1}$ (Figure \ref{fig:1}).	
}
\subsection{Axis-Crossing Number}
In this section we show how to distinguish between path configurations on the basis of
the number of crossings with the axis. We also show that repositioning the
axis has no utility.

We follow the notation in \cite{RSvi}. Throughout this subsection, assume that the sources and sinks are fixed on the two faces $f_1,f_2$.  Let $C_1,C_2$ to be the two circles denoting the two faces $f_1$ and $f_2$ respectively. Hence $C_1$ is outside $C_2$ in the embedding of $G$. Let $P$ be a set of $k-$disjoint paths; then define the Axis-Crossing$(P)$ as the number of times the paths in $P$ cut the $P_{axis}$ in the clockwise direction minus the number of times they cut $P_{axis}$ in the anti-clockwise direction. Here, a path is always directed from a source to a sink. Denote $\AC(P)=\operatorname{Axis-Crossing}(P)$.
We generalise \cref{lemma:rs} slightly for arbitrary \emph{path configutrations}.
\begin{lemma}\label{lemma:axise}
    For a fixed set of sources and sinks with exactly $\ell_1$ sinks on $f_1$ and $\ell_2$ sources on $f_2$, for any path configuration $\mcP$ where every matching is between a source and a sink, and $P,Q\in\mcP$, $\AC(P)\equiv \AC(Q)\equiv O_{\mcP}\bmod q$, where $O_{\mcP}\in\bbZ_q$ is a constant only dependent on $\mcP$.
\end{lemma}
Note that the above Lemma also holds when $\mcP\in\mcM_{q'}$ even when $q'>q$. The proof is pushed to \cref{appendix:1}. For a path configuration $\mcP$, define its Axis-Crossing to be $O_{\mcP}$, as above.
\begin{lemma}\label{lemma:Lemma axis}
	 For a fixed set of sources and sinks with exactly $\ell_1$ sinks on $f_1$ and $\ell_2$ sources on $f_2$, and for any two configurations $\mcP\ne\mcQ$ where every matching is between a source and a sink such that $\mcPi=\mcQi$ and $|\mcPi|=\ell_i$ for $i=1,2$, we have $O_{\mcP}\ne O_{\mcQ}$.
\end{lemma}
The proof is deferred to \cref{appendix:1}. For the following Lemma, assume there are two axes $P_{axis}^1,P_{axis}^2$, and the axis-crossing of a configuration $\mcP$ w.r.t the $i^{th}$ axis is $O_{\mcP}^i$ for $i=1,2$.
\begin{lemma}\label{lemma:Lemma 5}
	 For a fixed set of sources and sinks with exactly $\ell_1$ sinks on $f_1$ and $\ell_2$ sources on $f_2$, and for any two configurations $\mcP,\mcQ$ where every matching is between a source and a sink, we have $O_{\mcP}^1-O_{\mcP}^2=O_{\mcQ}^1-O_{\mcQ}^2$ i.e. axis-crossing of all configurations change by a fixed constant if we change the position of the axis.
\end{lemma}
\begin{proof}
	First, assume there is only one terminal (wlog $z^{(1)}_1$) on $f_1$ between $P^2_{axis}$ and $P^1_{axis}$ (in the clockwise sense), and no such terminal on $f_2$ between the 2 axis, as in Fig \ref{fig:2(a)}. Fix a $\mcP$, let $p$ be the path of $\mcP$ with one of the endpoints in $z^{(1)}_1$.
	\begin{itemize}
		\begin{figure}[hbt!]
			
			\begin{minipage}[c]{0.5\linewidth}
				\centering
				\includegraphics[width=7cm]{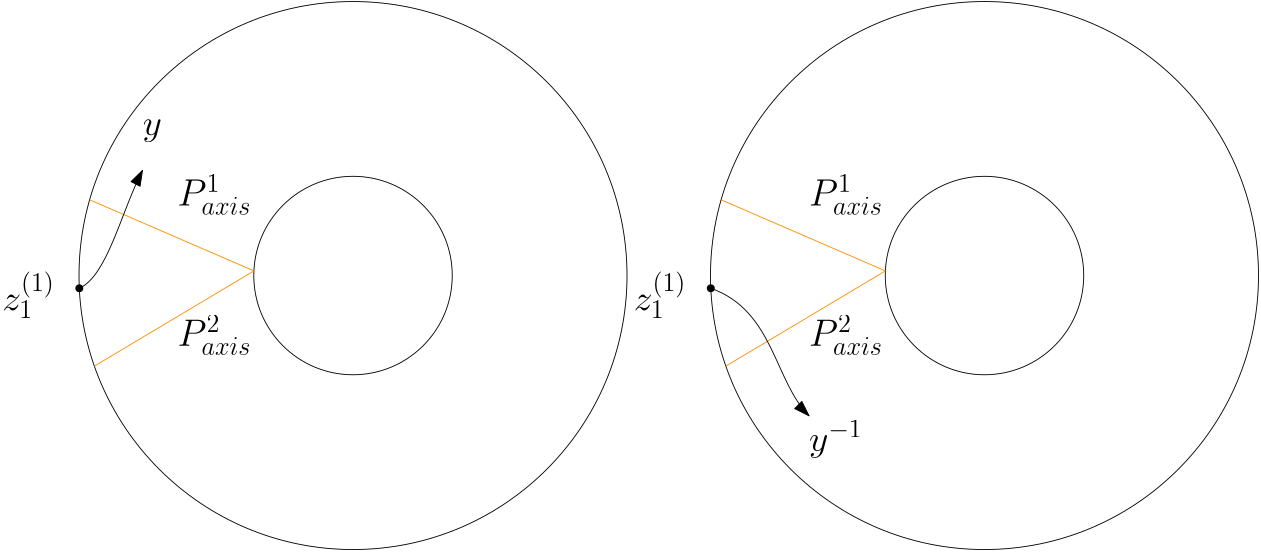}
				\caption{$z^{(1)}_1$ is a source.}
				\label{fig:2(a)}
			\end{minipage}
				\begin{minipage}[c]{0.5\linewidth}
				\centering
				\includegraphics[width=7cm]{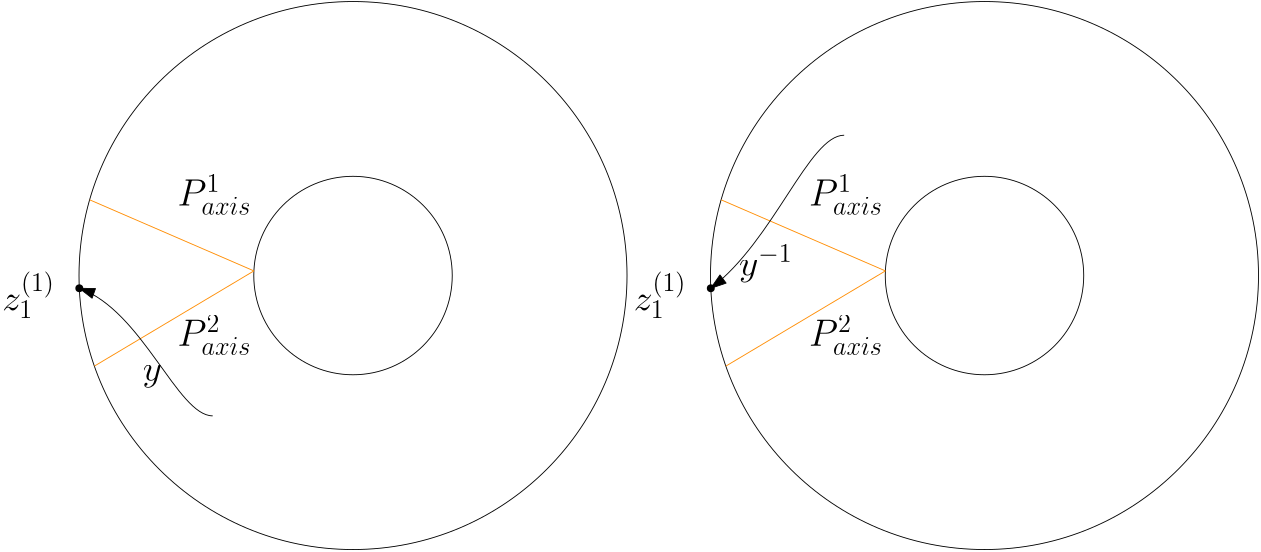}
				\caption{$z^{(1)}_1$ is a sink.}
				\label{fig:2(b)}
			\end{minipage}
		\end{figure}
		\item [1.] $z^{(1)}_1$ is a source. Then $p$ cuts $P^1_{axis}$ $1+$ the number of times it cuts $P^2_{axis}$ (Fig \ref{fig:2(a)}). Hence $O^1_{\mcP}-O^2_{\mcP}=1$.
		\item [2.] $z^{(1)}_1$ is a sink. Then $p$ cuts $P^1_{axis}$ $-1+$ the number of times it cuts $P^2_{axis}$ (Fig \ref{fig:2(b)}). Hence $O^1_{\mcP}-O^2_{\mcP}=-1$.
	\end{itemize}
	Hence, $O^1_{\mcP}-O^2_{\mcP}$ only depends on whether $z^{(1)}_1$ is a sink or source, and not $\mcP$. The same proof holds when there is exactly one terminal between the two axis on $f_2$. Now, it is easy to inductively generalise this for two arbitrary axis $P^1_{axis},P^2_{axis}$ since we can consider a chain of axis such that every two consequitive axis have exactly one terminal between them on a face, and each time the change is $O_{\mcP}$ is the same. Hence, for any $\mcP,\mcQ$ we have $O_{\mcP}^1-O_{\mcP}^2=O_{\mcQ}^1-O_{\mcQ}^2$.
\end{proof}

\subsection{The Algorithm}
\label{subsec:algo}
In this section, we describe the algorithm  for the ``odd\footnote{For $k$ terminals on both faces then putting all sinks on one face, sources on the other and setting axis-crossing to $0$ solves the case when all the paths are across the two faces, irrespective of parity of $k$.} two-face'' $k$-$\SDPP$. Recall that $K_i$ is the set of terminals on $f_i$.

	Let $J_1,J_2$ be a set of distinct terminals on $f_1,f_2$ respectively, such that $|J_i|=\ell_i$. Here, $\ell_1,\ell_2$ are parameters dependent on $q$. Define $G(J_1,J_2)$ to be the preprocessing step on $G$ such that \begin{itemize}
		\item $J_1$ is the set of sinks on $f_1$ and $K_1\setminus J_1$ is the set of sources on $f_1$.
		\item $J_2$ is the set of sources on $f_2$ and $K_2\setminus J_2$ is the set of sinks on $f_2$.
		
	\end{itemize} Also, $\perm(J_1,J_2)$ is defined as $\perm(A_{G(J_1,J_2)})$ where each entry of $A_{G(J_1,J_2)}$ is multiplied with $y^q$ so that monomials in the permanent are in $\bbZ[x,y]$ instead of $\bbZ[x](y)$ and modulo $q$ the exponents of $y$ are the same. We shall vary $J_1,J_2$ in our algorithm.

Call a \textit{path configuration} $\mcP$ compatible with $G(J_1,J_2)$ if the matchings in $\mcP$ have one endpoint in a sink and the other in a source in $G(J_1,J_2)$. The following follows from \cref{fact:mat}.
\begin{proposition}\label{prop:comp}
	 $\perm(J_1,J_2)=\sum\limits_{\mcP}h_{\mcP}(x)$ where $\mcP$ is compatible with $G(J_1,J_2$).
	\end{proposition}
	 For given $J_1,J_2$ as above, call $\mcM_q(J_1,J_2,\tau)$ as the set of \textit{path configurations} $\mcP$ with $O_{\mcP}=\tau(\operatorname{mod}q)$ that are compatible with $G(J_1,J_2)$. From \cref{lemma:axise}, we know all disjoint path instances of the same configuration come with the same $\AC\mod q$. Hence
 \[[y^\tau]_q\perm(J_1,J_2)=\sum\limits_{\mcP\in\mcM_q(J_1,J_2,\tau)}h_{\mcP}(x).\label{eq:E1}\tag{E1}\]

Let $\{\mcP_i\}_i$ be the set of all \textit{path configurations}. Let $\bmM$ be the matrix whose rows are indexed\footnote{ Note that we also vary $q$ i.e. the sizes of $J_1,J_2$.} by the tuples $(J_1,J_2,\tau)$ where $|J_i|=\ell_i\le k_i/2,q=(k_i-\ell_i)/2,\tau\in\bbZ_q$ and the columns of $\bmM$ are indexed by $\{\mcP_i\}_i$. Define
\[
\bmM[(J_1,J_2,\tau),\mcP_i]=
\begin{cases}
	1 &\text{ if } \mcP_i\in \mcM_q(J_1,J_2,\tau) \nonumber \\
	0 &\text{ otherwise}
\end{cases}
\]
Define the \textit{path vector} $\bmv=(h_{\mcP_1}(x),h_{\mcP_2}(x),\ldots)^T$, and the \textit{permanent vector} $\bmp$ where the $(J_1,J_2,\tau)^{\text{th}}$ coordinate is $[y^\tau]_qP(J_1,J_2)$. Hence, restating \eqref{eq:E1} we have 
\[\bmM\cdot\bmv=\bmp\label{eq:E2}\tag{E2}\]
This brings us to the following two lemmas whose proofs are pushed to \cref{appendix:1}.
\begin{lemma}\label{lem:c1}
	$\bmM$ is a square matrix with dimensions less than $k'\times k'$ where $k'\le  k4^k$.
\end{lemma}
\begin{lemma}\label{lem:comp}
	$\bmM$ is an invertible matrix. $\operatorname{Det}(\bmM)< 2^{f_0(k)}$ where $f_0(k)=(2k^2\operatorname{log}k)2^{2k}$. $\operatorname{Adj}(\bmM)$ can be computed in a constant size, depth circuit.
\end{lemma}
\begin{algorithm}[t]
	\caption{Odd Two-Face $k$-$\SDPP(\mcP)$}\label{algo:1}
	\begin{algorithmic}[1]
		\scriptsize
		
		\item[1.] Give random weights $r_e \in [4n^2]$ to the lower order bits of the original graph $G$; i.e. edge $e$ gets weight $4w_en^2+r_e$ where $r_e$ are the random weights.
		\item[2.] $\forall q,\ell_i=(k_i-q)/2,\tau\in\bbZ_q$ and $J_i\in \binom{K_i}{\ell_i}$, compute $[y^\tau]_q\perm(J_1,J_2)\mod 2^{f_0(k)}$ i.e entries of $\bmp\bmod 2^{f_0(k)}$.
		\item[3.] Compute $H_{\mcP}=\det(\bmM) h_{\mcP}\bmod 2^{f_0(k)}$ using the identity $\det(\bmM)\bmv=\operatorname{Adj}(\bmM)\cdot \bmp$ and step 2.
		\item[4.] The least degree monomial $x^w$ in $H_{\mcP}(x)$ corrsponds to the unique shortest disjoint path instance under isolation.
		\item[5.] Perform standard decision to search in parallel. See for example \cite{DIKM}.
		
	\end{algorithmic}
\end{algorithm}
Thus we reach the proof of the main theorem with the aid of \cref{algo:1}:
\begin{proof}[Proof of \cref{theorem:1}]
	Let $\sF\subset 2^{[m]}$ be the set of all disjoint path instances of $\mcP$. Then by our choice of random weights $r_e$ and Isolation Lemma (\cref{lemma:iso}), the shortest disjoint path instances of $\mcP$ w.r.t the random weights is unique with probability at least $3/4$. Further the higher order bits given by $4w_en^2$ assure that this instance is indeed a shortest disjoint path instance on $G$.
	
	We find $\bmp\bmod {2^{f_0(k)}}$ by computing the entries $[y^{\tau}]_q\perm(J_1,J_2)\bmod {2^{f_0(k)}}$ using \cref{prop:extract}. Hence, step 2 and 3 of the Algorithm can be computed in $\NC^2$. We can compute from \cref{lem:comp} the $\operatorname{Adj}(\bmM)$ in $\AC^1$. Also $\det(\bmM)<2^{f_0(k)}$ from our choice of $f_0$ justifying the correctness of step 4. This completes the proof.
\end{proof}

\section{Proof of Correctness}
\label{sec:proof}
In this section, we prove that $\bmM$ is invertible. We massage the equations \eqref{eq:E2} so that they become triangular with non-zero diagonal. We do this by left multiplying $\bmM$ with a square matrix $\bmL$ to cancel out certain unwanted \emph{configurations}. 
%
\paragraph*{Setting up Notation}
We set up the machinery to solve the equations given by \eqref{eq:E2}. Let $\mcM_q$ be the set of all path configurations with exactly $q$ paths going across $f_1$ and $f_2$. We also number the terminals on the face $f_i$ by $\{0,1,\ldots,k_i-1\}$ and always draw the $P_{axis}$ is the following way: On $f_1$ it is between $k_1-1$ and $0$, and on $f_2$ it is between $k_2-1$ and $0$. Realise that the labels ($z_i^{(1)},z_i^{(2)}$), but we vary this numbering only to specify the positioning of the $P_{axis}$ i.e. we vary the $P_{axis}$. Recall from \cref{lemma:Lemma 5} that varying the $P_{axis}$ will not generate new information. But in this part of the proof, we vary the position of the $P_{axis}$ to reduce bookkeeping. 
		\begin{figure}[H]
		
		\begin{minipage}[c]{\linewidth}
			\centering
			\includegraphics[width=10cm]{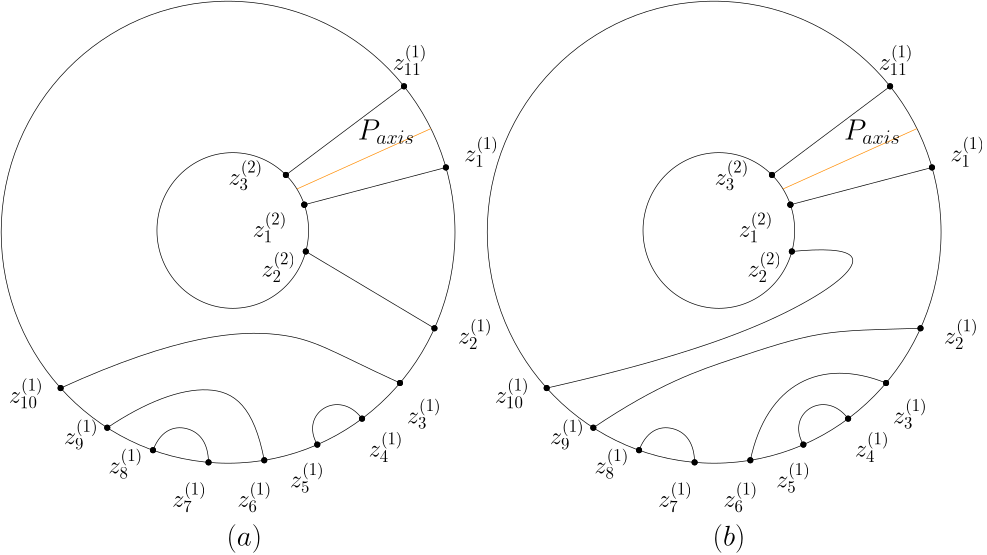}
			\caption{$\A=(z^{(1)}_5,z^{(1)}_8)$ in both $(a),(b)$. In $(a)$, $\N=(2,2)$, whereas in $(b)$, $\N=(3,1)$. $\alpha^{(1)}(\mcP)=\{z^{(1)}_3,z^{(1)}_4,z^{(1)}_6,z^{(1)}_7\}$ and $\{z^{(1)}_2,z^{(1)}_3,z^{(1)}_4,z^{(1)}_7\}$ in $(a)$ and $(b)$ respectively. $\Rp^{(1)}$ in both figures is $\{z^{(1)}_8\}$.}
			\label{fig:4(a)}
		\end{minipage}
		
	\end{figure}
Recall from \cref{sec:3.1}, for a path between $\alpha_p$ and $\beta_p$, they are respectively the clockwise first and last terminals in their \emph{even interval}.
Define $\Phi:\mcM_q \to \binom{K_1}{t\le \ell_1}\times \bbN^{t}\times \binom{K_2}{t'\le\ell_2}\times \bbN^{t'}$ as follows: 

$\Phi(\mcP)=((\A,\N),(\B,\M))$ where $\A$ be the set of \emph{pivots} on $f_1$ given by $\{\beta_p|p\in\mcPo,\alpha_p+1=\beta_p\}$ and $\sa_1<\ldots<\sa_t$ in the order of the labels i.e. $z^{(1)}_1,\ldots ,z^{(1)}_{k_1}$.
$\N=(\n_1,\ldots,\n_t)$ such that $\alpha^{(1)}(\mcP)=\cup_{i\in[t]}\{\sa_i-1,\sa_i-2,\ldots,\sa_i-\n_i\}$. In other words, $n^{(1)}_i$ is the maximum number of $\alpha_p$'s that preceed $\sa_i$
We say that the $\n_i$ paths that have an endpoint in $\{\sa_i-1,\ldots,\sa_i-\n_i\}$ \emph{contribute} to $n_i$ and \emph{correspond} to $\sa_i$.\label{contribute}
Similarly define $(\B,\M)$ such that $\alpha^{(2)}(\mcP)=\bigcup_{i\in[t']}\{\sab_i-1,\sab_i-2,\ldots,\sab_i-\m_i\}$.
Refer to Fig \ref*{fig:4(a)} for an example. Observe that $(\A,\N),(\B,\M)$ uniquely determine $\alpha^{(1)}(\mcP),\alpha^{(2)}(\mcP)$, hence $\Phi$ restricted to $\mcPi$ for $i=1,2$ is an injection.

For a \textit{path configuration} $\mcP$, define $\wt \Phi(\mcP)=(\aA,\bB)$ as follows:\\ If $\Phi(\mcP)=((\A,\N),(\B,\M))$, where $\A=(\sa_1,\ldots,\sa_t),\N=(\n_1,\ldots,\n_t)$, then $\aA$ is the sequence where $\sa_i$ occurs $\n_i$ times i.e $\aA=(\underbrace{\sa_1,\ldots,\sa_1}_{\n_1},\ldots,\underbrace{\sa_t,\ldots,\sa_t}_{\n_t})$. Define $\bB$ similarly w.r.t $(\B,\M)$. Observe that if $\mcP\in \mcM_q$, then $|\aA|=\ell_1=(k_1-q)/2,|\bB|=\ell_2=(k_2-q)/2$. For a fixed $q$, let $J_1=(j_1,\ldots,j_{\ell_1})\in \binom {K_1}{\ell_1}$ be the sequence of terminals designated to be sinks on $f_1$, and let $J_2=(j_1',\ldots,j_{\ell_2}')\in \binom {K_2}{\ell_2}$ be the sequence of terminals designated to be the sources on $f_2$. Note that $q$ is not fixed. We say that $\aA_r$ is the pivot of $j_r$, and let $(\aA_r-1,\aA_r)$ be its \emph{pivot pair}. Then define $\delta$ as follows:
\[
\text{For $j_r\in J_1$, } \delta(j_r)=\begin{cases}0 &  j_r\in \{\aA_r,\aA_r+1,\ldots,k_1-1\}\\
	1 & j_r\in \{0,\ldots,\aA_r-1\}\end{cases}\text{, here $r$ varies in $[\ell_1]$.} \label{eq:delta1}\tag{1}
\]
\[\text{For $j_r'\in J_2$, } \delta(j_r')=\begin{cases}0 & j_r'\in \{\bB_r,\bB_r+1,\ldots,k_2-1\}\\
	1 & j_r'\in \{0,\ldots,\bB_r-1\}\end{cases}\text{, here $r$ varies in $[\ell_2]$.} \label{eq:delta2}\tag{2}
\]
In other words, for $j_r\in J_1$, $\delta(j_r)=1$ \emph{iff} $j_r$ is strictly between the axis and its pivot $\aA_r$ in the clockwise sense. And $ \delta(J_i)=\sum\limits_{j\in J_i} \delta(j) \text{ for $i=\{1,2\}$}$ is the number of sinks(resp. sources) on $f_1$(resp. $f_2$) which are between the axis their pivot in the clockwise order. Also, assume that $0$ is a \textit{free} vertex\footnote{This the second place where we assume that there are odd number of terminals on both faces.} in $\mcP$ on either face. If not, we can wlog renumber the vertices accordingly. Observe that renumber only changes the position of the $P_{axis}$. Since the $P_{axis}$ is adjacent to a \emph{free} vertex, only the paths across the two faces in $\mcP$ contribute to axis-crossing.

Now we define the following polynomial:
\[
\mcF_{\mcP}(x)=\sum\limits_{I_1,I_2,J_1,J_2}(-1)^\sigma\left ([y^{O_{\mcP}+\tau}]_q\perm(J_1,J_2)\right ) \text{ where we vary $I_i\in \bbZ_q^{\ell_i},J_i\in \binom {K_i}{\ell_i}$}\label{eq:E3}\tag{E3}\]\[
\text{where } \tau(I_1,I_2,J_1,J_2,\aA,\bB)=2\Sigma_{I_1}+2\delta(J_1)-2\Sigma_{I_2} -2\delta(J_2)\text{, and}\]\[
\sigma(I_1,I_2,J_1,J_2,\aA,\bB)=\Sigma_{I_1}+\Sigma_{I_2}+\Sigma_{J_1}+\Sigma_{J_2}+\delta(J_1)+\delta(J_2) 
\]
Here we use the convention $\Sigma_S=\Sigma_{s\in S}s$.
Observe that equations \eqref{eq:E1} and \eqref{eq:E3} implie
\[\mcF_{\mcP}(x)=\sum\limits_{I_1,I_2,J_1,J_2}(-1)^\sigma\sum\limits_{\mcQ\in\mcM_q(J_1,J_2,\tau+O_{\mcP})}h_{\mcQ}(x)\label{eq:E4}\tag{E4}\] 
Now, consider the following matrix $\bmL$ whose rows are indexed by $\{\mcP_i\}_i$ and columns by the tuples $(J_1,J_2,\tau)$.
\[
\bmL[\mcP_i,(J_1,J_2,O_{\mcP_i}+\tau)]=	\sum\limits_{\vI,\vJ,\vA:\tau(\vI,\vJ,\vA)=\tau}(-1)^\sigma  \label{eq:E5}\tag{E5}
\] where $(\vI,\vJ,\vA)=(I_1,I_2,J_1,J_2,\aA,\bB)$. Clearly, $\bmL$ is motivated by the equation \eqref{eq:E3}. Now, consider the matrix $\bmF=\bmL\cdot\bmM$ whose rows and columns are both indexed by $\{\mcP_i\}_i$. This matrix is motivated by the equation \eqref{eq:E4}. Realise that \eqref{eq:E2} is now transformed to $\bmF\cdot\bmv=\bmL\cdot\bmp$. We prove that $\bmF$ is triangular with nonzero diagonal, hence invertible; this implies that $\bmM$ is invertible. 
\subsection{Cancelling Unwanted Path Configurations}\label{sec:cancel}
In this section, we show that unwanted path configurations are cancelled out in pairs in $\bmF=\bmL\cdot\bmM$. These path configurations are divided into two categories, namely \emph{bad} configurations and $f_1,f_2-$equivalent. If $h_{\mcQ}(x)$ appears with nonzero coefficient in $\mcF_{\mcP}(x)$, then we say that $\mcQ\in\mcF_{\mcP}$.

\subparagraph*{Cancelling Bad Configurations}
Fix a \textit{path configuration} $\mcP$. Say, $\wt \Phi(\mcP)=(\aA,\bB)$. Let $\mcQ\in\mcM_q(J_1,J_2,\tau+O_{\mcP})$ for some $J_1,J_2,\tau$. Let $J_1=(j_1,\ldots,j_{l_1})$ be the set of sinks on $f_1$, and $p$ be a path in $\mcQ$ with both its endpoints on $f_1$ with one of the endpoints at $j_r$ i.e. $p$ is the path on $f_1$ ending at $j_r$. Then we say $p$ is a $good$ path if the \emph{pivot pair} $(\aA_r-1,\aA_r)$ of $j_r$ lies in the even interval of $\alpha_p,\beta_p$. Otherwise, we say $p_r$ to be a $bad$ path. We make an analogous definition for $f_2$.Note that a \textit{bad} path always has both its endpoints on the same face. We call a path configuration $\mcP$ to be \textit{bad} if it contains some \textit{bad} path. Note that badness is w.r.t some fixed \textit{path configuration} $\mcP$.
	\begin{figure}[hbt!]
	\begin{minipage}[c]{0.5\linewidth}
		\centering
		\includegraphics[width=6.5cm]{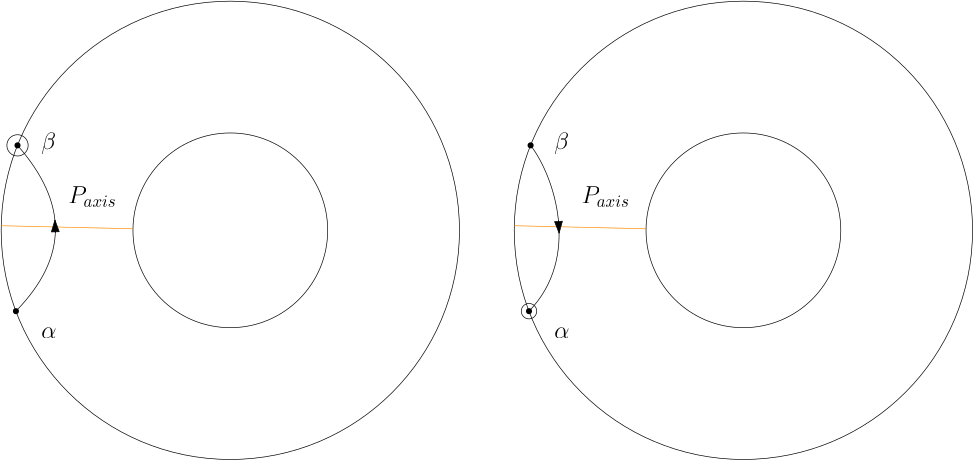}
		\caption{The circled terminal is the sink which is switched from $J_1$ to $\wt J_1$.}
		\label{fig:3(a)}
	\end{minipage}
	\begin{minipage}[c]{0.5\linewidth}
		\centering
		\includegraphics[width=6.5cm]{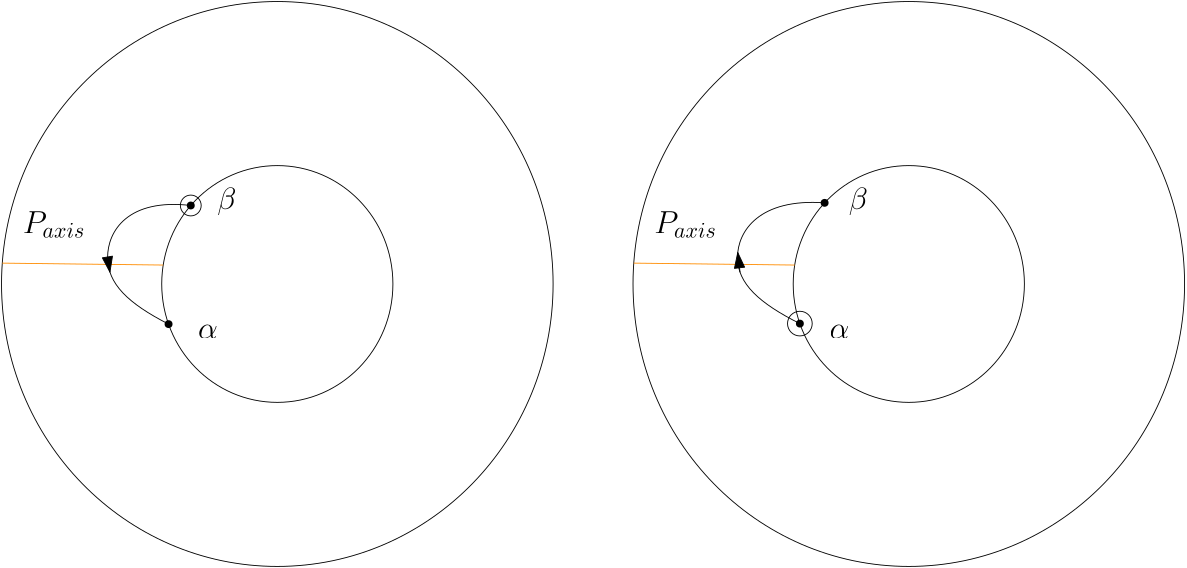}
		\caption{The circled terminal is the source switched from $J_2$ to $\wt J_2$.}
		\label{fig:3(b)}
	\end{minipage}
\end{figure}
%
The following Cancellation Lemma allows us to cancel out all \emph{bad} configurations in pairs.
\begin{lemma}[Cancellation Lemma 1]\label{lemma:bcancel}
	For all $\mcP$, all bad path configurations $\mcQ$ cancel out in pairs in $\mcF_{\mcP}(x)$ i.e. $\bmF[\mcP,\mcQ]=0$ for all bad configurations $\mcQ$.
\end{lemma}
\begin{proof}
	Suppose we get a path configuration $\mcQ\in\mcM_q(J_1,J_2,\tau+O_{\mcP})$ with some \textit{bad path}, and $I_1=(i_1,\ldots,i_{l_1}),J_1=(j_1,\ldots,j_{l_1})$, $I_2=(i_1',\ldots,i_{l_2}'),J_2=(j_1',\ldots,j_{l_2}')$. Say that $\mcP$ has the paths $(p_1,\ldots,p_{l_1})$, $(p_1',\ldots,p_{l_2}')$ ending at and starting at $J_1,J_2$ respectively in the same order, such that $p_r$ is the first path that is \textit{bad}. Say the endpoints of $p_r$ are $j_r$ and $\wt j_r$. Wlog assume $j_r\prec \wt j_r$. Consider $\wt J_1=(j_1,\ldots,j_{r-1},\wt j_r,j_{r+1},\ldots,j_{l_1}),\wt I_1=I_1,\wt I_2=I_2,\wt J_2=J_2$ where we replace $j_r$ with $\wt j_r$ in $J_1$ to yield $J_2$.
	\begin{itemize}
		\item \texttt{Case 1:} $p_r$ cuts the $P_{axis}$. In this case, $k_1-1,0$ lies in the even interval of $j_r$ and $\wt j_r$ i.e. parity of $j_r$ is the same as that of $\wt j_r$ (this is because $k_1-1$ and $0$ have the same parity\footnote{This is the third place we assume both face have an odd number of terminals}). Since $p_r$ is \textit{bad}, $j_r\in \{\aA_r,\ldots,k_1-1\}$ and $\wt j_r\in \{0,\ldots,\aA_r-1\} $, therefore $\delta(\wt J_1)=\delta(J_1)+1$ from definition \eqref{eq:delta1}. This implies that the new values\footnote{Here $\wt \tau,\wt \sigma$ are used as shorthand for $\tau(\wt I_1,\wt I_2,\wt J_1,\wt J_2),\sigma(\wt I_1,\wt I_2,\wt J_1,\wt J_2)$.} $\wt \tau=\tau+2,\wt \sigma=\sigma+1$. Since the path $p_r$ cuts the axis, the axis-crossing of $\mcQ$ at $(I_1,J_1,I_2,J_2)=2+\mbox{ axis-crossing at }(\wt I_1,\wt J_1,\wt I_2,\wt J_2)$ (using \cref{claim:ax2}). Hence, $ \mcQ\in\mcM_q( J_1, J_2, \tau)$ and $\mcM_q(\wt J_1,\wt J_2,\wt \tau)$ with opposite signs $(-1)^{\sigma}$ and $(-1)^{\wt\sigma}$.
		\item \texttt{Case 2:} $p_r$ does not cut the $P_{axis}$. In this case, $k_1-1,0$ does not lie in the even region of $j_r$ and $\wt j_r$ i.e. parity of $j_r$ is not equal to that of $\wt j_r$. Also either both $j_r,\wt j_r\in \{\aA_r,\ldots,k_1-1\},\text{ or } \{0,\ldots,\aA_r-1\}$ i.e. $\delta(J_1)=\delta(\wt J_1)$. This implies that $\wt \tau=\tau,\wt \sigma=\sigma+1$. Since $p_r$ does not cut $P_{axis}$, this means that $ \mcQ\in\mcM_q( J_1, J_2, \tau)$ and $\mcM_q(\wt J_1,\wt J_2,\wt \tau)$ with opposite signs $(-1)^{\sigma}$ and $(-1)^{\wt\sigma}$.
	\end{itemize}
	Note that the above shows an involution between the monomials in $\mcF_{\mcP_0}(x)$ with some \emph{bad} paths with a positive sign and those with a negative sign.
	If all the paths $p_r$ were $good$, then take the first $p_r'$ which is a $bad$ path. In which case take $\wt I_1=I_1,$ $\wt J_1=J_1,\wt I_2=I_2,\wt J_2=(j_1',\ldots,j_{r-1}',\wt j_r',j_{r+1}',\ldots,j_{l_2}')$ where $p_r$ has endpoints $j_r'\prec \wt j_r'$. Note that a \textit{bad} path always has both its endpoints on a single face. Here, we can argue in the case analogous to case 1 above, $\wt \tau=\tau-2,\wt\sigma=\sigma-1$ using a result similar to \cref{claim:ax2} with the aid of \cref{fig:3(b)}, we establish an involution. The case 2 is virtually identical in both.

\end{proof}
Thus we are done modulo the following claim.
\begin{claim}\label{claim:ax2}
	$\mcQ\in\mcM_q(\wt J_1,J_2,\tau+2)$.
\end{claim}
\begin{claimproof}
	Observe from \cref{fig:3(a)} that when the sink is on $\alpha$, $p$ cuts the axis $+1$ times, whereas when the sink is on $\beta$, $p$ cuts the axis $-1$ times.. Hence, $\mcQ$ is also compatible with $G(\wt J_1,J_2)$, but with axis-crossing = $2+$ that in $G(J_1,J_2)$ (i.e increases by 2).
\end{claimproof}
We will use the following corollary in the next subsection. 
\begin{corollary}\label{cor:bad}
	For configurations $\mcP,\mcQ$ such that $\bmF[\mcP,\mcQ]\ne 0$, let $\Phi(\mcP)=((\A,\N),(\B,\M))$, $\Phi(\mcQ)=((\wt\A,\wt\N),(\wt\B,\wt\M))$. Then we have $\wt\A\subseteq\A,\wt\B\subseteq\B$.
\end{corollary}
\begin{proof}
	 Suppose $\exists z\in\A\setminus\wt\A$. Therefore is a path in $\wt \mcP$ with endpoints in $z-1$ and $z$.
	The even interval of $z-1,z$ does not contain any other terminal. Since $z\notin \A$, or equivalently, $z\notin\aA$, this even interval does not contain any \emph{pivot pair}. Hence such a path is bad by definition. A similar proof follows for $\B$.
\end{proof}

\subparagraph*{Cancelling $f_1,f_2$-equivalent Configurations}
Call two \emph{path configurations} $\mcP,\mcQ$ to be $f_1,f_2-$\emph{equivalent} if $\mcPi=\mcQi$ for $i=1,2$ i.e. the set of paths with both endpoints on the same face is the same in $\mcP$ and $\mcQ$.
\begin{lemma}[Cancellation Lemma 2]\label{lemma:faceeq}
	Let $\mcP\ne\mcQ$ be $f_1,f_2-$equivalent configurations. Then $\mcQ$ cancels out in $\mcF_{\mcP}(x)$ i.e. $\mcQ\notin\mcF_{\mcP}(x)$ or equivalently $\bmF[\mcP,\mcQ]=0$. Also, $\bmF[\mcP,\mcP]\ne 0$.
\end{lemma}
\begin{proof}
	Let $\mcP,\mcQ$ be such \textit{path configurations}. We will show that $\mcQ$ gets canceled in $\mcF_{\mcP}(x)$. Note that the axis-crossing of $\mcQ$ is not $O_{\mcP}$ (since otherwise $\mcP=\mcQ$ from \cref{lemma:axise}). Say $\mcQ\in \mcM_q(J_1,J_2,\tau+O_{\mcP})$ for some $\tau\ne 0$ ; where $I_1=(i_1,\ldots,i_{l_1})$, $J_1=(j_1,\ldots,j_{l_1}),I_2=(i_1',\ldots,i_{l_2}')$, $J_2=(j_1',\ldots,j_{l_2}')$, and the paths due to $J_1,J_2$ are $(q_1,\ldots,q_{l_1}),(q_1',\ldots,q_{l_2}')$. Say the end points of $q_r$ are $\alpha_r\prec \beta_r$, and that for $q_r'$ are $\alpha_r'\prec \beta_r'$. Note that $\delta(\alpha_r)=1,\delta(\beta_r)=0$ since $\mcQ$ is \textit{good} w.r.t $\mcP$ and therefore $\alpha_r,\beta_r$ encloses the corresponding \emph{pivot pair} $(\aA_r-1,\aA_r)$ and therefore $\alpha_r\in \{0,1,\ldots,\aA_r-1\}$ and $\beta_r\in\{\aA_r,\ldots,k_1-1\}$. 
	
	Since $\tau\ne 0\mod q$, which implies there exists some $q_r$ (or $q_r'$; wlog assume $q_r$) such that either $j_r=\alpha_r$ and $i_r\ne q-1$, or $j_r=\beta_r$ and $i_r\ne 0$ (since otherwise it can be shown that $\tau=0\mod q$). Let $q_r$ be the first such path. Wlog assume $j_r=\alpha_r,i_r\ne q-1$, and set $\wt j_r=\beta_r,\wt i_r=i_r+1$. Set $\wt I_1=(i_1,i_2,\ldots,i_{r-1},\wt i_r,i_{r+1},\ldots,i_{l_1}),\wt J_1=(j_1,j_2,\ldots,j_{r-1},\wt j_r,j_{r+1},\ldots,j_{l_1}),\wt I_2=I_2,\wt J_2=J_2$. Observe that $\delta(\wt j_r)=0,\delta(j_r)=1$, therefore $\delta(\wt J_1)=\delta(J_1)-1$. Thus, we have $\wt \tau=\tau$, which implies that $\mcQ$ also occurs due to $\wt I_1,\wt J_1,\wt I_2,\wt J_2$. Now all we have to show is that it occurs with the opposite sign. Since $k_1-1,0$ do not lie in the even interval of $\wt j_r,j_r$ (as $0$ is a free terminal), $j_r$ and $\wt j_r$ have opposite parity. Further the parity of $\delta(\wt J_1),\delta(J_1)$ is opposite and so is that of $\wt i_r,i_r$ 
	 . Therefore $\wt \sigma$ and $\sigma$ have the opposite parity.
	 
	  We need to show that $\mcP$ appears with a nonzero coefficient in $\mcF_{\mcP}(x)$. Consider a term in $\mcF_{\mcP}(x)$ arising from $I_1=(i_1,\ldots,i_{l_1})$, $J_1=(j_1,\ldots,j_{l_1}),I_2=(i_1',\ldots,i_{l_2}')$, $J_2=(j_1',\ldots,j_{l_2}')$. Each sink/source corresponds to some path of $\mcP$. Let $\mcPo=(p_1,\ldots,p_{l_1}),\mcPt=(p_1',\ldots,p_{l_2}')$ in the order of $J_1,J_2$ i.e. sinks and sources. Say $p_r$ has endpoints $\alpha_r\prec \beta_r$, and $p_r'$ has endpoints $\alpha_r'\prec\beta_r'$. We know that if for all $r$, $j_r\in\{\alpha_r,\beta_r\},j_r'\in\{\alpha_r',\beta_r'\}$ and $\tau=0$, then $\mcP$ is compatible with the current term $\perm(J_1,J_2)$ (from \cref{prop:comp})because the endpoints and axis-crossing match. We know that if for any index $r$, either $j_r=\alpha_r$ and $i_r\ne q-1$, or $j_r=\beta_r$ and $i_r\ne 0$, the term gets canceled in $\mcF_{\mcP}(x)$ (from the above proof). And observe that in the remaining two cases, namely, either $j_r=\alpha_r$ and $i_r=q-1$ or $j_r=\beta_r$ and $i_r=0$, we have $2i_r+2\delta(j_r)=0\mod q$, i.e. these terms have $\tau=0$; also note that $i_r+j_r+\delta(j_r)$ (contribution to $\sigma$) has the same parity, hence these terms appear with the same sign. Therefore $ \mcP\in\mcF_{\mcP}(x)$. Notice that $\bmF[\mcP,\mcP]$ is a power of two. (Note how we used the fact that $q-1\text{ and }0$ have the same parity which is only true when both faces have an odd number of terminals\footnote{This is the fourth place we assume both the faces have an odd number of terminals}).
\end{proof}
\subsection{Invertibility of $\bmM$}\label{sec:inv}
We now show that $\bmM$ is invertible  by showing that $\bmF$ is triangular
with non-zero diagonal and hence is invertible.
For a configuration $\mcP$, define $\Rp^{(1)}(\mcP)\subseteq \A(\mcP)$ to be the set of \emph{rightmost pivots} of $\mcPo$ as follows: $\sa_i\in\Rp^{(1)}(\mcP)$ if for each paths $p$ in $\mcPo$ enclosing $\sa_i$ in its even interval, it is the clockwise rightmost pivot enclosed by $p$. Refer to \cref{fig:4(a)} for an example. Similarly define $\Rp^{(2)}(\mcP)\subseteq \B(\mcB)$ w.r.t $\mcPt$. Here  $\A(\mcP),\B(\mcP)$ and $\N(\mcP),\M(\mcP)$ are from $\Phi(\mcP)$.
\begin{lemma}\label{lemma:rp}
For configurations $\mcP,\mcQ\in\mcM_q$ with a fixed $q$, if $\bmF[\mcP,\mcQ]\ne 0$, then either:
\begin{itemize}
	\item $\A(\mcQ)\subseteq\A(\mcP),\B(\mcQ)\subseteq\B(\mcP)$ with at least one of the inclusions being strict, barring which
	\item $\Rp^{(1)}(\mcQ)\subseteq\Rp^{(1)}(\mcP),\Rp^{(2)}(\mcQ)\subseteq\Rp^{(2)}(\mcP)$ with at least one of the inclusions being strict, barring which
	\item $\forall$ \emph{rightmost pivots} $ \sa_i\in\Rp^{(1)}(\mcQ)$, $\n_i(\mcQ)\le\n_i(\mcP)$ and $\forall \sab_i\in\Rp^{(2)}(\mcQ)$ $\m_i(\mcQ)\le\m_i(\mcP)$.
\end{itemize}
\end{lemma}	
\begin{proof}
	  From \cref{cor:bad}, it is clear that $\A(\mcQ)\subseteq\A(\mcP),\B(\mcQ)\subseteq\B(\mcP)$.	Assume that both are equal.
	
 Let $\A=\A(\mcP)=\A(\mcQ)=\{\sa_1,\ldots,\sa_r\}$ for some $r$. Therefore, for all pivots $\sa_i$ of $\mcP$ and therefore of $\mcQ$, there is a path from $\sa_i-1$ to $\sa_i$ in $\mcQ$. Now we show $\Rp^{(1)}(\mcQ)\subseteq\Rp^{(1)}(\mcP)$. Note that $\mcQ$ must be a good configuriation from \cref{lemma:bcancel}.  Say there is a pivot $\sa_i\in\A(\mcP)\setminus\Rp^{(1)}(\mcP)$ which is a \emph{rightmost pivot} in $\mcQo$. Since $\sa_i$ is not a \emph{rightmost pivot} in $\mcP$, $\exists$ a path $p$ which encloses $\sa_i$ and $\sa_{i+1}$ where $\sa_{i+1}$ is the pivot clockwise right of $\sa_i$. This path corresponds to some $\sa_j$ where $\sa_j$ is clockwise left to $\sa_i$. Since $\sa_i$ is a \emph{rightmost pivot} in $\mcQ$, the paths that correspond to $\{\sa_j,\sa_j+1,\ldots,\sa_i\}$ do not enclose $\sa_{i+1}$ because $\sa_{i+1}$ is to the right of $\sa_i$. But one of these paths must enclose $\{\sa_j,\sa_j+1,\ldots,\sa_i\}$ in their even interval from \emph{good}ness of $\mcQ$. Hence at least one of the $\beta$ endpoints of these paths (say a path $p$) includes a terminal from $\{\sa_{i+1}-2,\ldots,\sa_{i+1}-\n_{i+1}\}$. But $\n_{i+1}$ paths have to enclose the pivot pair $\{\sa_{i+1}-1,\sa_{i+1}\}$ in their even interval (this follows not from the \emph{good}ness of $\mcP$, but from the \emph{good}ness of $\mcQ$). But the number of terminals strictly between $\beta_p$ and $\sa_{i+1}$ is strictly less than $\n_i$. Hence there exists a path that has its $\alpha$ endpoint to the left of $\alpha_p$, i.e. this path encloses both $\sa_j$ and $\sa_{i+1}$, and hence also $\sa_i$. But $\sa_{i+1}$ is rightward of $\sa_i$, thus $\sa_i$ cannot be a rightmost pivot in $\mcQo$ and we reached a contradiction. A similar proof holds for the face $f_2$. Therefore we assume  $\Rp^{(1)}(\mcQ)\subseteq\Rp^{(1)}(\mcP),\Rp^{(2)}(\mcQ)\subseteq\Rp^{(2)}(\mcP)$.	Hence, assume that both are equal.
	
 For the third point, assume that for some terminal $\sa_i\in\A$, we have $\n(\mcQ)>\n(\mcP)$. Therefore some path corresponding to some $\sa_j$ in $\mcP$ contributes to $\n_i(\mcQ)$ i.e corresponds to $\sa_i$ in $\mcQ$ . Here $\sa_j$ is to the clockwise right of $\sa_i$ since otherwise this path would have been \emph{bad}. Since this corresponds to $\sa_j$ in $\mcP$, from \emph{good}ness of $\mcQ$, we know that the path encloses $\sa_i$ and $\sa_j$ in $\mcQ$. Hence, $\sa_i$ is not a rightmost pivot in $\mcQ$ and so in $\mcP$ (since they have same set of righmost pivots). An analogous proof holds for $f_2$. Therefore, $\forall$ \emph{rightmost pivots} $ \sa_i\in\Rp^{(1)}(\mcQ)$, $\n_i(\mcQ)\le\n_i(\mcP)$ and $\forall \sab_i\in\Rp^{(2)}(\mcQ)$ $\m_i(\mcQ)\le\m_i(\mcP)$.
\end{proof}
\begin{lemma}\label{lemma:inv}
	\begin{itemize}
		\item[(i)] For all configurations $\mcP$, $ \bmF[\mcP,\mcP]\ne 0$
		\item[(ii)] For distinct configurations $\mcP_0,\mcP_1,\ldots,\mcP_t$, if $\bmF[\mcP_{i-1},\mcP_{i}]\ne 0\forall i\in[t]$, then $\bmF[\mcP_t,\mcP_0]=0$.
	\end{itemize}
\end{lemma}
\begin{proof}
The proof of $(i)$ is direct from \cref{lemma:faceeq}. Next we prove (ii). Assume to the contrary that $\bmF[\mcP_t,\mcP_0]\ne0$. Let $\mcP_i\in\mcM_{q_i}$ for some $q_i$'s. From \cref{prop:q}, $q_0\le q_1\le\ldots\le q_t\le q_0$. Hence, all the $q_i$'s are equal. From \cref{lemma:faceeq}, we know that $\mcPi\ne \mcQi$ for at least one of the faces $f_1$ or $f_2$,  since otherwise they are $f_1,f_2$-equivalent. Also, from \cref{lemma:rp} we observe that $\A(\mcP_0)\subseteq\A(\mcP_t)\subseteq\ldots\A(\mcP_0),\B(\mcP_0)\subseteq\B(\mcP_t)\subseteq\ldots\B(\mcP_0),\Rp^{(1)}(\mcP_0)\subseteq\Rp^{(1)}(\mcP_t)\subseteq\ldots\Rp^{(1)}(\mcP_0),\Rp^{(2)}(\mcP_0)\subseteq\Rp^{(2)}(\mcP_t)\subseteq\ldots\Rp^{(2)}(\mcP_0)$. Therefore, all of the inclusions must be equality. Let the corresponding sets be $\A,\B,\Rp^{(1)},\Rp^{(2)}$ respectively. We also know that for all rightmost pivots $\sa_i\in\Rp^{(1)}$, $\n_i(\mcP_0)\le\n_i(\mcP_t)\le\ldots\n_i(\mcP_0)$ forcing all of them to be equality. Therefore all the configurations $\mcP_0,\ldots,\mcP_t$ include paths with endpoints in $\{(\sa_i-1,\sa_i),\ldots,(\sa_i-\n_i,\sa_i+\n_i-1)\}$ for all $\sa_i\in \Rp^{(1)}$. Observe that this means, in each configuration, there are exactly $\n_i$ paths corresponding to $\sa_i$ for the rightmost pivots. Hence, we can remove the terminals $\{\sa_i-n_i,\ldots,\sa_i-1,\sa_i,\ldots,\sa_i+\n_i-1\}$ from the chain of configurations, and apply structural induction -- In each inductive step the number of terminals, and the number of paths with both endpoints on $f_1$ go down. We are forced to get a contradiction at some point since, if in each inductive step, we remove the same set of terminals, then $\mcPo_0,\mcPo_1,\ldots,\mcPo_t$ remain equal after the removal. The same argument can be applied to the face $f_2$. Thus, $\mcPi_0=\ldots=\mcPi_t$ for $i=1,2$. But then from \cref{lemma:faceeq}, we know that $\mcP_0=\mcP_1=\ldots=\mcP_t$. for at least one of the faces $f_1$ or $f_2$,  since otherwise they are $f_1,f_2$-equivalent. Hence, we reach a contradiction. 
\end{proof}
\begin{corollary}
	$\bmF$ is invertible.
\end{corollary}
\begin{proof}
	Consider the directed simple graph $\mcG$ where the vertices are path configurations. $(\mcP,\mcQ)$ is an edge in $\mcG$ if $\bmF[\mcP,\mcQ]\ne0$. From \cref{lemma:inv}(ii) we know that $\mcG$ is a directed acyclic graph. Order the rows and columns of $\bmF$ by a topological sort on the vertices of $\mcG$. Hence, it is clear that $\bmF$ is triangular and \cref{lemma:inv}(i) implies that the diagonal elements are non zero. Hence, $\bmF$ is invertible.
\end{proof}

\section{Shortest $(A+B,q)$ paths}
\label{sec:abpaths}
In this section, we solve $(A+B,q)$-$\SDPP$. Let $A$ and $B$ be a set of terminals with $|A|=k_1,|B|=k_2$. The goal is to find the shortest disjoint paths where exactly $q$ paths have one endpoint in $A$ and one in $B$, and the rest of the paths are among $A$ terminals or $B$ terminals. Setting $q=0$ gives us the shortest $A+B$ paths as defined in \cite{HiraiNamba}. We always have $k_1\equiv k_2\equiv q\mod 2$ and $q\le \min\{k_1,k_2\}$. Wlog, assume $k_1\le k_2$. Let $\mcM_t$ be the set of path configurations with exactly $t$ paths across $A$ and $B$. And let $\mcH_t(x)=\sum\limits_{\mcP\in\mcM_t}h_{\mcP}(x)$.

    Let $J_1\in \binom A {t_1}, J_2\in \binom B {t_2} $ and\footnote{This is because $\#$-sinks=$\#$-sources.} $t_1+k_2-t_2=t_2+k_1-t_1$.  Define $G(J_1,J_2)$ to be the preprocessing step on $G$ such that:
        (1) $J_1$ is the set of sinks on $A$ and $A\setminus J_1$ is the set of sources on $A$.
        (2) $J_2$ is the set of sources on $B$ and $B\setminus J_2$ is the set of sinks on $B$.
\begin{lemma}\label{lemma:ab}
   $\forall t\sum_{J_1\in\binom A {t_1},J_2\in\binom B {t_2} }\perm(J_1,J_2)=\sum_{i\ge t}2^{(k-2i)/2}\mcH_i(x)\text{ for $t_i=\frac {k_i-t}2$ $i=1,2$.}$
\end{lemma}
\begin{proof}
   Let $\mcP\in \mcM_i$ for some\footnote{Notice, $i\equiv k_1 \equiv k_2\mod 2$. } $i\ge t$.
For each path in $\mcP$ with both endpoints on $A$, we can vary the sink on either endpoint leading to $2^{(k_1-i)/2}$ ways. Similarly, we can vary the endpoints of the sources for the paths with both endpoints on $B$ leading to $2^{(k_2-i)/2}$. And for the paths across $A$ and $B$, the sinks are always on $B$. Therefore we get $2^{(k_1-i)/2}\cdot 2^{(k_2-i)/2}=2^{(k-2i)/2}$ choices of $J_1\in\binom A {t_1},J_2\in\binom B {t_2}$ that lead to $\mcH_i(x)$. Since there are $t_1=(k_1-t)/2$ sinks in $A$, the number of paths with both terminals in $A$ cannot be more than $t_1$. Hence the number of paths across $A$ and $B$ is at most $t$. Thus, for $i<t$, $\mcM_i$ type \textit{path configurations} do not occur in the left sum since they are not compatible with $\G(J_1,J_2)$.
\end{proof}
\begin{algorithm}[t]
	\caption{$(A+B,q)$-$\SDPP$}\label{algo:2}
	\begin{algorithmic}[1]
		\scriptsize
		
		\item[1.] Give random weights $r_e \in [4n^2]$ to the lower order bits of the original graph $G$; i.e. edge $e$ gets weight $4w_en^2+r_e$ where $r_e$ are the random weights.
		\item[2.] $\forall t\ge q,t_i=(k_i-t)/2$ and $J_i\in \binom{K_i}{t_i}$, compute $\perm(J_1,J_2)\mod 2^{k+1}$ using \cref{lemma:dj}.
		\item[3.] Inductively compute $2^{(k-2t)}\mcH_t(x)\mod 2^{k+1} \forall t\ge q$.
		\item[4.] The least degree monomial $x^w$ of $2^{(k-2q)}\mcH_t(x)\mod 2^{k+1}$ corresponds to the unique shortest $(A+B,q)$ path instance under isolation.
		\item[5.] Perform standard decision to search in parallel. See for example \cite{DIKM}.
		
	\end{algorithmic}
\end{algorithm}
\begin{proof}[Proof of \cref{theorem:2}]
	We prove the correctness of \cref{algo:2}. Observe from \cref{lemma:ab} that we can inductively compute $2^{(k-2t)}\mcH_t(x)\mod 2^{k+1} \forall t\ge q$. This is because \cref{lemma:ab} can be re-written as $2^{(k-2t)/2}\mcH_t(x)=\sum_{J_1\in\binom A {t_1},J_2\in\binom B {t_2} }\perm(J_1,J_2)-\sum_{i>t}2^{(k-2i)/2}\mcH_{i}(x)$ where $t_i=\frac {k_i-t}2$ $\forall t$. The probability calculations remain similar to proof of \cref{theorem:1}.
\end{proof}

\section{Conclusion}
\label{sec:concl}
We give an algorithm for the Odd Two-face $k$-$\SDPP$  that is in $\RNC$.
We would like to point out that the oddness requirement is crucial for our
techniques as we have indicated in the proofs. The first and most important
open question is eliminating this condition -- we believe that a new idea
is required to accomplish this.
We also parallelise the shortest $(A+B)$-$\SDPP$ introduced by Hirai and Namba
\cite{HiraiNamba} by reducing it to inverting a matrix as well -- in fact we
solve a slighly generalised version of the original problem. This sidesteps
the attempts to parallelise the Hafnian modulo a power of $2$ and that still
remains open from \cite{DJ}.

\bibliographystyle{plainurl}
\bibliography{main}

\appendix
\section{Appendix}
\subsection{Left out proofs from \cref{sec:3}}\label{appendix:1}
The following are the left out proofs:
\begin{proof}[Proof of \cref{lemma:axise}]
	For a path $p$ in $P$ that has both its endpoints on the same face, say $f_1$, its contribution to the Axis-Crossing is uniquely determined by whether $P_{axis}$ passes through the even interval defined by its endpoints and the direction of $p$ (which depends on the setting of sources and sinks that is fixed). This is because the region enclosed by $p$ and $f_1$ (that does not contain $f_2$) has to contain an even number of terminals, and $p$ either cuts the $P_{axis}$ in the clockwise or counter-clockwise direction or does not cut $P_{axis}$ at all depending on its endpoints.
	
	Let $L_P,L_Q$ be the set of paths in $P$ and $Q$ respectively across the two faces. Now in $\mcP$ say, there are $u_1$ paths on $f_1$, and $u_2$ paths on $f_2$. Therefore, there are $\ell_1-u_1$ sinks and $k_1-\ell_1-u_1=\ell_1-u_1+q$ sources on $f_1$ that are matched to $L_P$ and $L_Q$. \begin{property}\label{prop:star}
		Therefore $\ell_1-u_1$ paths are directed from $f_2$ to $f_1$, and $\ell_1-u_1+q$ paths are directed from  $f_1$ to $f_2$ in $L_P$ and $L_Q$
	\end{property}Let $\wnd(L_Q)-\wnd(L_P)=c\in \bbZ$. The paths from $f_1$ to $f_2$ in $Q$ have $\AC$ = $c+$ $\AC$ of the paths from $f_1$ to $f_2$ in $P$, since the winding number is measured clockwise from $C_1$ to $C_2$. Similarly, the paths from $f_2$ to $f_1$ in $Q$ have $\AC$ = $-c+$ $\AC$ of the paths from $f_2$ to $f_1$ in $P$.  Therefore, using \cref{prop:star}, we have $\AC(Q)=c((\ell_1-u_1+q)-(\ell_1-u_1))+\AC(P)\equiv \AC(P) \text{ }(\operatorname{mod}q)$. Therefore, $\AC(P)\equiv O_{\mcP}\text{ } (\operatorname{mod} q)$ for some $O_{\mcP}\in\bbZ_q$.
\end{proof}
\begin{proof}[Proof of \cref{lemma:Lemma axis}]
	Observe that $\mcPi=\mcQi$ implies that the contribution to $\AC$ due to $\mcPi,\mcQi$ are the same. Hence we have to show that the contribution to $\AC$ due to the $q$ paths across $f_1$ and $f_2$ in $\mcP,\mcQ$ are different $(\bmod q)$. This is follows from \cite[Lemma 21]{DIKM}. 
	
\end{proof}
\begin{proof}[Proof of \cref{lem:c1}]
	We show a bijection between the columns (path configurations) and rows (tuples of the form $(J_1,J_2,\tau)$). Take any path configuration $\mcP$. Define $J_i=\{\alpha_p|p\in\mcPi\}=\alpha^{(1)}(\mcP)$ as the \emph{alpha sequences}, and $\tau=O_{\mcP}$ when the set of sinks on $f_1$ and sources on $f_2$ are fixed to $J_1,J_2$ respectively. Also, note that specifying $\alpha^{(1)}(\mcP)=J_1,\alpha^{(2)}{\mcP}=J_2$ uniquely identifies $\mcPo,\mcPt$, and setting $O_{\mcP}=\tau$ uniquely determines the paths across $f_1$ and $f_2$ in $\mcP$ (from \cref{lemma:Lemma axis}). This establishes the required bijection. Hence, the matrix is square. 
	
	Now, the number of rows in the matrix is the number of ways to choose $q\le k$, and $J_i\in\binom{K_i}{\ell_i}$, which is $\le k2^{k_1}2^{k_2}=k4^k$.
\end{proof}
\begin{proof}[Proof of \cref{lem:comp}]
	Invertibility of $\bmM$ follows from \cref{sec:proof}. From \cref{lem:c1}, $\operatorname{Det}(\bmM)\le(k4^k)!< 2^{f_0(k)}$.
	
	First, observe that computing $\bmM$ can be done in a constant size circuit. For every $(J_1,J_2,\tau)$ and $\mcP$ check whether $\mcP$ is compatible with $G(J_,J_2)$ and its axis-crossing is $\tau$ --if yes, $\bmM[(J_1,J_2,\tau),\mcP]=1$, and $0$ otherwise, and there are only $k4^k\times k4^k$ such checks. Now that we have computed $\bmM$, it is easy to compute ($\Log$ computable, in fact in constant depth) $\operatorname{Adj}(\bmM)$ since $\bmM$ is a constant ($k4^k$) size matrix.
\end{proof}
\subsection{Example}\label{ex:one}
In this section we show with an example that the matrix $\bmM$ may not be triangular even by permuting its rows and columns. Consider the example with $3$ terminals on one face and $1$ on the other. Notice that we do not need the $P_{axis}$ in this example since there can only be one path across the two faces, and the coefficient of $[y^0]_1$ will be the same polynomial.
\begin{figure}[hbt!]
	
	\begin{minipage}[c]{\linewidth}
		\centering
		\includegraphics[width=12cm]{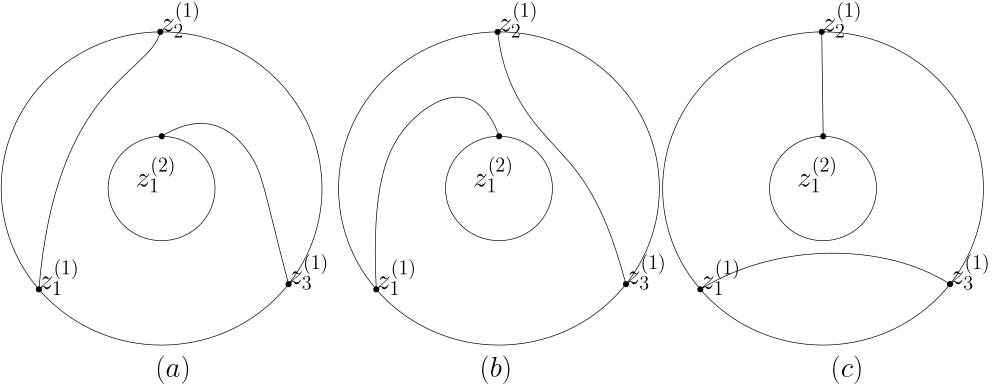}
		\caption{The configurations are $\mcP_1$,$\mcP_2$ and $\mcP_3$ respectively.}
		\label{fig:3v1}
	\end{minipage}
	
\end{figure}

Now we write the equations $\perm(J_1,J_2)$ by varying $J_i$ on the face $f_i$ (see \eqref{eq:E1}) as the set of sinks and sources on $f_1,f_2$ repectively.

\begin{minipage}{.4\textwidth}
	\begin{align*}
		\perm(\{z^{(1)}_1\},\{\})=h_{\mcP_1}(x)+h_{\mcP_3}(x)\\
		\perm(\{z^{(1)}_2\},\{\})=h_{\mcP_1}(x)+h_{\mcP_2}(x)\\
		\perm(\{z^{(1)}_3\},\{\})=h_{\mcP_2}(x)+h_{\mcP_3}(x)
	\end{align*}
\end{minipage}
\begin{minipage}{.6\textwidth}
\begin{align*}
	\bmM=
	\begin{blockarray}{cccc}
		& \mcP_1 & \mcP_2 & \mcP_3  \\
		\begin{block}{c[ccc]}
			\left(\{z^{(1)}_1\},\{\}\right) & 1 & 0 & 1 \\
			\left(\{z^{(1)}_2\},\{\}\right) & 1 & 1 & 0 \\
			\left(\{z^{(1)}_3\},\{\}\right) & 0 & 1 & 1 \\
		\end{block}
	\end{blockarray}
\end{align*}
\end{minipage}

Observe that $\bmM$ is invertible but not triangular. Now we write the equations corresponding to \eqref{eq:E3},\eqref{eq:E4}:

\begin{align*}
	\mcF_{\mcP_1}(x)&=\perm(\{z^{(1)}_1\},\{\})+\perm(\{z^{(1)}_2\},\{\})-\perm(\{z^{(1)}_3\},\{\})=2h_{\mcP_1}(x)\\
	\mcF_{\mcP_2}(x)&=-\perm(\{z^{(1)}_1\},\{\})+\perm(\{z^{(1)}_2\},\{\})+\perm(\{z^{(1)}_3\},\{\})=2h_{\mcP_2}(x)\\
	\mcF_{\mcP_3}(x)&=\perm(\{z^{(1)}_1\},\{\})-\perm(\{z^{(1)}_2\},\{\})+\perm(\{z^{(1)}_3\},\{\})=2h_{\mcP_3}(x)
\end{align*}
\begin{minipage}{.4\textwidth}
\begin{align*}
	\bmF=
	\begin{blockarray}{cccc}
		& \mcP_1 & \mcP_2 & \mcP_3  \\
		\begin{block}{c[ccc]}
			\mcP_1 & 2 & 0 & 0 \\
			\mcP_2 & 0 & 2 & 0 \\
			\mcP_3 & 0 & 0 & 2 \\
		\end{block}
	\end{blockarray}
\end{align*}
\end{minipage}
\begin{minipage}{.6\textwidth}
 Observe that $\bmF$ is triangular, in fact a diagonal matrix because of the simplicity in the particular example. In general, in may not even be triangular w.r.t the given ordering, hence we need to carefully choose an order to make it triangular.
\end{minipage}
\subsection{Extracting bivariate coefficients modulo $2^k$}
We want to compute the coefficients of a bivariate polynomal modulo powers of $2$. We do this by extending the univariate approach in \cite{DJ}. They consider the ring $\mfR=\bbZ[x]/(p(x))$ for some irreducible polynomial $p(x)$ in $\bbZ_2[x]$. From this they derive the ring $\mfR_k=\bbZ[x]/(2^k,p(x))$. In particular $\mfR_1=\bbZ_2[x]/(p(x))$ is a finite field of characteristic $2$. 

Our analogue of \cite[Lemmata~21,22]{DJ} is as follows:

\begin{lemma}
	Let $\bbF$ be a finite field of characteristic $2$, and order $q$.
	\[\sum_{a,b\in\bbF^*}a^{m_1}b^{m_2}=\begin{cases} 
	1 & \text{if }q-1\mid m_1\wedge q-1\mid m_2 \\
	0 & \text{ otherwise}
	\end{cases}
	\]
	Let $f(x,y)=\sum_{i,j=0}^dc_{i,j}x^iy^j$ be a polynomial with integer coefficients where the individual degrees are bounded by $d<q-1$. Then for any $t_1,t_2\le d$:
	\[\sum_{a,b\in\bbF^*}a^{q-1-t_1}b^{q-1-t_2}f(a,b)=c_{t_1,t_2} \bmod 2
	\]
\end{lemma}
\begin{proof}
	\[\sum_{a,b\in\bbF^*}a^{m_1}b^{m_2}=\sum_{a\in\bbF^*}a^{m_1}\sum_{b\in\bbF^*}b^{m_2}=1 \text{ if both the terms are equal to $1$, and $0$ otherwise.}
	\] 
	Hence, part 1 of the Lemma follows from \cite[Lemma 21]{DJ}. Moving on two the second identity:
	\[\sum_{a,b\in\bbF^*}a^{q-1-t_1}b^{q-1-t_2}f(a,b)=\sum_{i,j}^d\sum_{a\in\bbF^*}a^{q-1-t_1+i}\sum_{b\in\bbF^*}b^{q-1-t_2+j}c_{i,j}=c_{t_1,t_2}.\]
\end{proof}
From the above lemma, we can compute the coefficients of $f$ modulo $2$. Now we lift these coefficients to higher powers of two using analogous techniques as in \cite{DJ}.
\begin{lemma}\label{lemma:extract}
	\[\sum_{\vec a,\vec b\in (\bbF^*)^{2^{k-1}}} \left(\prod_{a\in\vec a}a\right)^{m_1}\left(\prod_{b\in\vec b}b\right)^{m_2}= \sum_{a,b\in\bbF^*}(a^{m_1}b^{m_2})^{2^{k-1}}=\begin{cases}
		1 \bmod 2^k& \text{if } q-1 \mid m_1\wedge q-1\mid m_2\\ 
		0 \bmod 2^k& \text{ otherwise}
	\end{cases}
	\]
	\[\sum_{\vec a,\vec b\in (\bbF^*)^{2^{k-1}}}\left(\prod_{a\in\vec a}a\right)^{q-1-t_1}\left(\prod_{b\in\vec b}b\right)^{q-1-t_2} f\left(\left(\prod_{a\in\vec a}a\right),\left(\prod_{b\in\vec b}b\right)\right)=c_{t_1,t_2}\bmod 2^k
	\]
	where the above computations are performed in the ring $\mfR_k$.
\end{lemma}
\begin{proof}
	This is a lift of the previous lemma from modulo $2$ to modulo $2^k$, the proof of which is analogous to that in \cite[Lemma 23]{DJ}. The above computations can be done in $\ParityL$ and hence in $\NC^2$ using \cite[Lemma 20]{DJ}. 
\end{proof}
See \cite{DJ} for an explicit realisation of the field $\bbF$. The above techniques generalise to constantly many variables as well.

\end{document}